\documentclass[12pt,draftcls,onecolumn]{IEEEtran}
\usepackage{amsmath,amssymb,amsthm}   
\usepackage{verbatim}   % useful for program listings
\usepackage{color}      % use if color is used in text
\usepackage{subfigure}  % use for side-by-side figures
\usepackage{tikz}
\newcounter{mytempeqncnt}
\usetikzlibrary{positioning}
\usepackage[english]{babel}
\usepackage{float}%\usepackage{keywords}

\usepackage{flushend}
\usepackage{verbatim}   % useful for program listings
\usepackage{color}      % use if color is used in text
\usepackage{subfigure}  % use for side-by-side figures
\usepackage{float}%\usepackage{keywords}
\usepackage{algorithm}
\usepackage{algorithmic}
\usepackage{setspace}
\usepackage{subfigure}
\usepackage{cite}
\newtheorem{lemma}{\bf Lemma}
\newtheorem{theorem}{\bf Theorem}
\newtheorem{proposition}{\bf Proposition}

\makeatletter
\newcommand*{\rom}[1]{\expandafter\@slowromancap\romannumeral #1@}
\makeatother

\begin{document}

\title{\textbf{Consensus based Detection in the Presence of Data Falsification Attacks}}
\author{Bhavya~Kailkhura,~\IEEEmembership{Student Member,~IEEE}, Swastik~Brahma,~\IEEEmembership{Member,~IEEE}, 
Pramod~K.~Varshney,~\IEEEmembership{Fellow,~IEEE}
\thanks{This work was supported by the Center for Advanced Systems and Engineering at Syracuse University.}
\thanks{The authors would like to thank Aditya Vempaty and Arun Subramanian for their valuable comments and suggestions to improve the quality of the paper.}
\thanks{B. Kailkhura, S. Brahma and P. K. Varshney are with Department of EECS, Syracuse University, Syracuse, NY 13244. (email: bkailkhu@syr.edu; skbrahma@syr.edu; varshney@syr.edu)}}

%\date{}
\maketitle

\begin{abstract} 
This paper considers the problem of detection in distributed networks in the presence of data falsification (Byzantine) attacks. Detection approaches considered in the paper are based on fully distributed consensus algorithms, where all of the nodes exchange information only with their neighbors in the absence of a fusion center. In such networks, we characterize the negative effect of Byzantines on the steady-state and transient detection performance of the conventional consensus based detection algorithms.  
To address this issue, we study the problem from the network designer's perspective. 
More specifically, we first propose a distributed weighted average consensus algorithm that is robust to Byzantine attacks. We show that, under reasonable assumptions, the global test statistic for detection can be computed locally at each node using our proposed consensus algorithm. We exploit the statistical distribution of the nodes' data to devise techniques for mitigating the influence of data falsifying Byzantines on the distributed detection system.
Since some parameters of the statistical distribution of the nodes' data might not be known \textit{a priori}, we propose learning based techniques to enable an adaptive design of the local fusion or update rules. 
\end{abstract}
\begin{keywords}
distributed detection, consensus algorithms, data falsification attacks, Byzantines
\end{keywords}

\section{Introduction}
Distributed detection is a well studied topic
in the detection theory literature \cite{Varshney, Viswanathan, veer}. 
The traditional distributed detection framework comprises of a group of spatially distributed nodes which acquire the observations regarding the phenomenon of interest and send them to the fusion center (FC) where a global decision is made. 
However, in many scenarios a centralized FC may not be available or in large networks, the FC can become an information bottleneck that may cause degradation of system performance, and may even lead to system failure. 
Also, due to the distributed nature of future communication networks, and various practical constraints, e.g., absence of the FC, transmit power or hardware constraints and dynamic characteristic of wireless medium, it may be desirable to employ alternate peer-to-peer local information exchange in order to reach a global decision.
One such distributed approach for peer-to-peer local information exchange and inference is the use of a consensus algorithm~\cite{paper2}. 

Recently, distributed detection based on consensus algorithms has been explored in~\cite{paper4,paper5,paper6,paper7,paper8,paper9}. In consensus based detection approaches, each node communicates only with its neighbors and updates its local state information about the phenomenon (summary statistic) by a local fusion rule that employs a weighted combination of its own value and those received from its neighbors. Nodes continue with this consensus iteration until the whole network converges to a steady-state value which is the global test statistic.
In particular, the authors in~\cite{paper7,paper8} considered average consensus based distributed detection and emphasized network designs based on the small world phenomenon for faster convergence~\cite{paper5}. A bio-inspired consensus scheme was introduced for spectrum sensing in~\cite{paper9}. However, these consensus-based fusion algorithms only ensure equal gain combining of local measurements. The authors in~\cite{paper4} proposed to use distributed weighted fusion algorithms for cognitive radio spectrum sensing. They showed that weighted average consensus based schemes outperform average consensus based schemes and achieve much better detection performance than the equal gain combining based schemes. However, the weighted average consensus based detection schemes are quite vulnerable to different types of attacks. 
One typical attack on such networks is a Byzantine attack. While Byzantine attacks (originally proposed in \cite{Lamport}) may, in general, refer to many types of malicious behavior, our focus in this paper is on  data-falsification attacks~\cite{frag, Rifa, Marano, Rawat, bhavyaj, Kailkhura2013, Kailkhura, aditya}.  
Thus far, research on detection in the presence of Byzantine attacks has predominantly focused on addressing these attacks under the centralized model~\cite{Marano,Rawat,paper11,aditya}.  
A few attempts have been
made to address the security threats in the distributed or consensus based schemes in recent research~\cite{tang,yu1,yu2,liu,yan,paper12}.
Most of these existing works on countering Byzantine or data falsification attacks in distributed networks rely on a threshold for detecting Byzantines. The main idea is to exclude nodes from neighbors list whose state information deviates significantly from the mean value.
In~\cite{yu1} and~\cite{yan}, two different defense schemes against data falsification attacks for distributed consensus-based detection were proposed. In~\cite{yu1}, the scheme eliminates the state value with the largest deviation from the local mean at each iteration step and, therefore, it can only deal with the situation in which only one Byzantine node exists. It excludes one state value even if there is no Byzantine node. In~\cite{yan}, the vulnerability of distributed consensus-based spectrum sensing was analyzed and an outlier detection algorithm with an adaptive threshold was proposed.
The authors in~\cite{liu} proposed a Byzantine mitigation technique based on adaptive local thresholds.
This scheme mitigates the misbehavior of
Byzantine nodes and tolerates the occasional
large deviation introduced by honest users. It
adaptively reduces the corresponding coefficients so that the Byzantines will eventually be isolated from the network. 

Excluding the Byzantine nodes from the fusion process may not be the best strategy from the network designer's perspective. As shown in our earlier work~\cite{aditya} in the context of distributed detection with one-bit measurements under a centralized model, an intelligent way to improve the performance of the network is to use the information of the identified Byzantines to the network's benefit. More specifically, learning based techniques have the potential to outperform the existing exclusion based techniques. In this paper, we pursue such a design philosophy in the context of raw data based fusion in decentralized networks.

To design methodologies for defending against Byzantine attacks,  fundamental challenges that arise are two-fold. First, how do nodes recognize the presence of attackers? Second, after identification of an attacker or group of attackers, how do nodes adapt their operating parameters?
Due to the large number of nodes and complexity of the distributed network, we develop and analyze schemes that would update their own operating parameters autonomously. Our approach further introduces an adaptive fusion based detection algorithm which supports the learning of the attacker's behavior. Our scheme differs from all existing work on Byzantine mitigation based on exclusion strategies~\cite{tang,yu1,yu2,liu,yan}, where the only defense is to identify and exclude the attackers from the consensus process. 

\subsection{Main Contributions} 
In this paper, we focus on the susceptibility and protection of consensus based detection algorithms. Our main contributions are summarized as follows:
\begin{itemize}
\item We characterize the effect of Byzantines on the steady-state performance of the conventional consensus based detection algorithms. More specifically, we quantify the minimum fraction of Byzantines needed to make the deflection coefficient of the global statistic equal to zero. 
\item Using probability of detection and probability of false alarm as measures of detection performance, we investigate the degradation of transient detection performance of the conventional consensus algorithms with Byzantines.
\item We propose a robust distributed
weighted average consensus algorithm and obtain closed-form expressions for optimal weights to mitigate the effect of data falsification attacks.
\item Finally, we propose a technique based on the expectation-maximization algorithm  and maximum likelihood estimation to learn the operating parameters (or weights) of the nodes in the network to enable an adaptive design of the local fusion or update rules.
\end{itemize}

The rest of the paper is organized as follows. In Sections~\ref{model} and \ref{attacks}, we introduce our system model and Byzantine attack model, respectively.
In Section~\ref{susceptibility}, we study the security performance of weighted average consensus based detection schemes. In Section~\ref{protection}, we propose a protection mechanism to mitigate the effect of data falsification attacks on consensus based detection schemes. Finally, Section~\ref{conclusion} concludes the paper.

\section{System model}
\label{model}

First, we define the network model used in this
paper. 
\subsection{Network Model}
\begin{figure}[t] 
  \centering
    \includegraphics[width=0.4\textwidth]{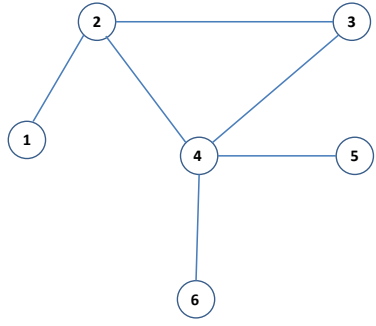}
    \caption{A distributed network with $6$ nodes}
    \label{fig1}
\end{figure}
 
We model the network topology as an undirected graph $ G=(V,E)$, where $V=\{v_{1},\cdots,v_{N}\}$ represents the set of nodes in the network with $|V|=N$.  
The set of communication links in the network correspond to the set of edges $E$, where $(v_{i},v_{j}) \in E$, if and only if there is a communication link between $v_i$ and $v_j$ (so that, $v_i$ and $v_j$ can directly communicate with each other). The adjacency matrix $ A $ of the graph is defined as
\[ a_{ij} = \left\{ \begin{array}{rll}
				1  & \mbox{if}\  (v_{i},v_{j}) \in E, \\
				0  & \mbox{otherwise.}
				\end{array}\right. 
\]
The neighborhood of a node $i$ is defined as
\begin{center}
$\mathcal{N}_{i}=\{v_{j} \in V : (v_{i},v_{j}) \in E\}, \forall i \in \{1,2,\cdots,N\}$.
\end{center}
The degree of a node $v_i$ in a graph $G$, denoted by $d_i$, is the number of
edges in $E$ which include $v_i$ as an endpoint, i.e., $d_i=\sum_{j=1}^Na_{ij}$.

The degree matrix $D$ is defined as a diagonal matrix with $\text{diag}(d_1,\cdots,d_N)$ and the Laplacian matrix $L$ is defined as  
\vspace{-0.05in}
\[ l_{ij} = \left\{ \begin{array}{rll}
				d_i  & \mbox{if}\ j=i,  \\
			   -a_{ij}  & \mbox{otherwise.}
				\end{array}\right. 
\]
or, in other words, $L=D-A$. For example, consider a network with six nodes trying to reach consensus (see Figure~\ref{fig1}). The Laplacian matrix $L$ for this network is given by
\begin{footnotesize}
\begin{equation*}
\begin{bmatrix}
   1 & -1& 0& 0& 0& 0 \\
   -1 & 3 & -1& -1& 0& 0 \\
    0 & -1& 2& -1& 0& 0 \\
    0 & -1& -1& 4& -1& -1 \\
    0 & 0& 0& -1& 1& 0 \\
    0 & 0& 0& -1& 0& 1
\end{bmatrix}.
\end{equation*}
\end{footnotesize}

The consensus based distributed detection scheme usually contains three phases: sensing, information fusion, and decision making. In the sensing phase, each node acquires the summary statistic about the phenomenon of interest. In this paper, we adopt the energy detection method so that the local summary statistic 
is the received signal energy. Next, in the information fusion phase, each node communicates with its neighbors to update their state values (summary statistic) and continues with the consensus iteration until the whole network converges to a steady state which is the global test statistic.
Finally, in the decision making phase, nodes make their own decisions about the presence of the phenomenon. Next, we describe each of these phases in more detail.

\subsection{Sensing Phase}
\label{Architecture}
We consider an $N$-node network using the energy detection scheme~\cite{paper16}. For the $i$th node, the sensed signal $x_{i}^k$ at time instant $k$ is given by

\[ x_{i}^k = \left\{ \begin{array}{rll}
				n_{i}^k,  & \mbox{under}\ H_{0} \\
			   h_{i}s^k+ n_{i}^k  &\mbox{under}\ H_{1}
				\end{array}\right. 
\]  

where $h_{i}$ is the channel gain, $s^k$ is the signal at time instant $k$, $n_{i}^k$ is AWGN, i.e., $n_{i}^k \sim N(0,\sigma_{i}^{2})$ and independent across time.  
Each node $i$ calculates a summary statistic $Y_{i}$ over a detection interval of $M$ samples, i.e.,
\begin{center}
$Y_{i}=\sum_{k=1}^M |x_{i}^k|^{2}$
\end{center} 
where $M$ is determined by the time-bandwidth product.
Since $Y_{i}$ is the sum of the square of $M$ i.i.d. Gaussian random variables, it can be shown that $\frac{Y_{i}}{\sigma_{i}^2}$ follows a central chi-square distribution with $M$ degrees of freedom $(\chi_{M}^2)$ under $H_{0}$, and, a non-central chi-square distribution with $M$ degrees of freedom and parameter $\eta_{i}$ under $H_{1}$, i.e.,
\[ \frac{Y_{i}}{\sigma_{i}^2} \sim \left\{ \begin{array}{rll}
				\chi_{M}^2,  & \mbox{under}\ H_{0} \\
			   \chi_{M}^2(\eta_{i})  &\mbox{under}\ H_{1}
				\end{array}\right. 
\]  
where
$\eta_{i}={E_{s}|h_{i}|^2}/{\sigma_{i}^2}$
is the local SNR at the $i$th node and 
$E_{s}=\sum_{k=1}^M |s^k|^{2}$
represents the sensed signal energy over $M$ detection instants. Note that the local SNR is $M$ times the average SNR at the output of the energy detector, which is $\frac{E_{s}|h_{i}|^2}{M\sigma_{i}^2}$. 

\subsection{Information Fusion Phase}
\label{algorithm}
Next, we give a brief introduction to conventional consensus algorithms~\cite{paper2}. Consensus is reached in two steps. 

Step 1: All nodes establish communication links with their neighbors, and broadcast their information state $x_{i}(0)=Y_i$. 
 
Step 2: Each node updates its local state information by a local fusion rule (weighted combination of its own value and those received from its neighbors). We denote node $i$'s updated information at iteration $k$ by $x_{i}(k)$. Node $i$ continues to broadcast information $x_{i}(k)$ and update its local information state until consensus is reached. 
This information state updating process can be written in a compact form as 
\begin{equation}
\label{consensusalgo}
x_{i}(k+1)=x_{i}(k)+\frac{\epsilon}{w_{i}}\underset{j\in \mathcal{N}_{i}}{\sum}(x_{j}(k)-x_{i}(k))
\end{equation}
where $\epsilon$ is the time step and $w_{i}$ is the weight assigned to node $i$'s information.
Using the notation $x(k)=[x_1(k),\cdots,x_N(k)]^T$, network dynamics in the matrix form can be represented as,
\begin{center}
$x(k+1)=Wx(k)$
\end{center}
where, $W=I-\epsilon\;\text{diag}(1/w_1,\cdots,1/w_N){L}$ is referred to as a Perron matrix. The consensus algorithm is nothing but a local fusion or update rule that fuses the nodes' local information state with information coming from neighbor nodes and every node asymptotically reaches the same information state for arbitrary initial values.  

\subsection{Decision Making Phase}
The final information state after reaching consensus for the above consensus algorithm will be the weighted average of the initial states of all the nodes~\cite{paper2} or $x_{i}^{*}={\sum_{i=1}^{N}w_{i}Y_i}/{\sum_{i=1}^{N}w_{i}}$, $\forall i$. Average consensus can be seen as a special case of weighted average consensus with $w_{i}=w,\;\forall i$.
After the whole network reaches
a consensus, each node makes its own decision
about the hypothesis using a predefined
threshold $\lambda$\footnote{\text{In practice, parameters such as threshold $\lambda$ and consensus time step $\epsilon$} can be set off-line. This study is beyond the scope of this work.}
\begin{center}
\[ \text{Decision} = \left\{ \begin{array}{rll}
				H_{1}  & \mbox{if}\ x_{i}^{*} > \lambda \\
			   H_{0}  & \mbox{otherwise}
				\end{array}\right.
\]
\end{center}
where weights are given by~\cite{paper4}
\begin{equation}
\label{convweight}
w_{i}=\dfrac{\eta_{i}/\sigma_{i}^2}{\sum_{i=1}^N{\eta_{i}/\sigma_{i}^2}}. 
\end{equation}
Note that, after reaching consensus $x_i^*=\Lambda,\forall i$. Thus, in rest the of the paper, $\Lambda$ is referred to as the final test statistic.

Next, we discuss Byzantine attacks on consensus based detection schemes and analyze the performance degradation of the weighted average consensus based detection algorithm due to these attacks. 

 \section{Attacks on Consensus based Detection Algorithms}
\label{attacks}
When there are no adversaries in the network, we noted in the last section that consensus can be accomplished to the weighted average of arbitrary initial values by having the nodes use the update strategy $x(k+1)=Wx(k)$ with an appropriate weight matrix $W$. Suppose, however, that instead of broadcasting the true sensing statistic $Y_i$ and applying the update strategy~\eqref{consensusalgo}, some nodes (referred to as Byzantines) deviate from the prescribed strategies.
Accordingly, Byzantines can attack in two ways: 
data falsification (nodes falsify their initial data or weight values) and consensus disruption (nodes do not follow update rule given by~\eqref{consensusalgo}).
More specifically, Byzantine node $i$ can do the following
\begin{eqnarray*}
&&\text{Data falsification:}\qquad\quad x_{i}(0)= Y_i+\Delta_i,\quad\text{or}\quad w_i\to \tilde{w}_i\\
&& \text{Consensus disruption:}\quad\;
x_{i}(k+1)= x_{i}(k)+\frac{\epsilon}{w_{i}}\underset{j\in \mathcal{N}_{i}}{\sum}(x_{j}(k)-x_{i}(k))+u_{i}(k),
\end{eqnarray*} 
where $(\Delta_i,\tilde{w}_i)$ and $u_i(k)$ are introduced at the initialization step and at the update step $k$, respectively. The attack model considered above is extremely general, and allows Byzantine node $i$ to update
its value in a completely arbitrary manner (via appropriate choices of $(\Delta_i,\tilde{w}_i)$, and $u_{i}(k)$, at each time step). 
An adversary performing consensus disruption attack
has the objective to disrupt the consensus operation.  
However, consensus disruption attacks can be easily detected because of the nature of the attack. The identification of consensus disruption attackers has been investigated in the past literature (e.g., see~\cite{sundar,paper5}) where control theoretic techniques were developed to identify disruption attackers in a single consensus iteration.  
Knowing the existence of such an identification mechanism, a smart adversary will aim to disguise itself while degrading the detection performance. 
In contrast to disruption attackers, data falsification attackers are more capable and can manage to disguise themselves while degrading the detection performance of the network by falsifying their data. 
Susceptibility and protection of consensus strategies to data falsification attacks has received scant attention, and this is the focus of our work. 
In this paper, we assume that an attacker performs only a data falsification attack by introducing $(\Delta_i,\tilde{w}_i)$ during initialization.
We exploit the statistical distribution of the initial values and devise techniques to mitigate the influence of Byzantines on the distributed detection system.

\subsection{Data Falsification Attack}
In data falsification attacks, attackers try to manipulate the final test statistic (i.e., $\Lambda=(\sum_{i=1}^{N}w_{i}Y_i)/(\sum_{i=1}^{N}w_{i})$) in a way that the detection performance is degraded. We consider a network with $N$ nodes that uses Algorithm~\eqref{consensusalgo} for reaching consensus. Algorithm~\eqref{consensusalgo} can be interpreted as, weight $w_{i}$, given to node $i$'s data $Y_i$ in the final test statistic, is assigned by node $i$ itself. So by falsifying initial values $Y_i$ or weights $w_i$, the attackers can manipulate the final test statistic. Detection performance will be degraded because Byzantine nodes can always set a higher weight to their manipulated information. Thus, the final statistic's value across the whole network will be dominated by the Byzantine node's local statistic that will lead to degraded detection performance. 

Next, we define a mathematical model for data falsification attackers. We analyze the degradation in detection performance of the network when Byzantines falsify their initial values $Y_i$ for fixed arbitrary weights $\tilde{w}_i$.

\subsection{Attack Model}  
The objective of Byzantines is to degrade the detection performance of the network by falsifying their data $(Y_i,w_i)$. 
By assuming that Byzantines are intelligent and know the true hypothesis, we analyze the worst case detection performance of the data fusion schemes. 
We consider the case when weights of the Byzantines have already been tampered to $\tilde{w}_i$ and analyze the effect of falsifying the initial values $Y_i$. 
This analysis provides the most favorable case from the point of view of Byzantines and yields the maximum performance degradation that the Byzantines can cause.
Now a mathematical model for a Byzantine attack is presented. Byzantines tamper with their initial values $Y_i$ and send $\tilde{Y_i}$ such that the detection performance is degraded. 

\noindent Under $H_0$: 
\[ \tilde{Y_{i}} =  \left\{ \begin{array}{rll}
				 Y_{i}+\Delta_{i}  & \mbox{with probability}\ P_{i} \\
			     Y_{i}  & \mbox{with probability}\ (1-P_i)\\
				\end{array}\right.  
\]
\noindent
Under $H_1$: 
\[ \tilde{Y_{i}} =  \left\{ \begin{array}{rll}
				 Y_{i}-\Delta_{i}  & \mbox{with probability}\ P_{i} \\
			     Y_{i}  & \mbox{with probability}\ (1-P_i)\\
				\end{array}\right.  
\]

where $P_{i}$ is the attack probability and $\Delta_{i}$ is a constant value which represents the attack strength, which is zero for honest nodes. As we show later, Byzantine nodes will use a large value of $\Delta_{i}$ so that the final statistic's value is dominated by the Byzantine node's local statistic that will lead to a degraded detection performance. 
We use deflection coefficient \cite{Kay} to characterize the security performance of the detection scheme due to its simplicity and its strong relationship with the global detection performance. Deflection coefficient of the global test statistic is defined as:
$\mathcal{D}(\Lambda)=\dfrac{(\mu_{1}-\mu_{0})^{2}}{\sigma_{(0)}^{2}}$, where $\mu_{k}=\mathbb{E}[\Lambda|H_k]$ is the conditional mean and $\sigma_{(k)}^2=\mathbb{E}[(\Lambda-\mu_{k})^2|H_k]$ is the conditional variance. The deflection coefficient is also closely related to other performance
measures, e.g., the Receiver Operating Characteristics (ROC) curve. In general, the detection performance monotonically increases with an increasing value of the deflection coefficient. We define the critical point of the distributed detection network as the minimum fraction of Byzantine nodes needed to make the deflection coefficient of global test statistic equal to zero (in which case, we say that the network becomes \textit{blind}) and denote it by $\alpha_{blind}$. We assume that the communication between nodes is error-free and our network topology is fixed during the whole consensus process and, therefore, consensus can be reached without disruption.

In the next section, we analyze the security performance of consensus based detection schemes in the presence of data falsifying Byzantines.
\begin{figure*}[]
\normalsize
\setcounter{mytempeqncnt}{\value{equation}}
\begin{equation}
\label{mu0}
\mu_{0}=\sum_{i=1}^{N_1}\left[P_i \dfrac{\tilde{w}_{i}}{\text{sum}(w)}(M\sigma_{i}^{2}+\Delta_{i})+(1-P_i)\dfrac{\tilde{w}_{i}}{\text{sum}(w)}(M\sigma_{i}^2)\right] 
+ \sum_{i=N_1+1}^{N}\left[\dfrac{w_{i}}{\text{sum}(w)}(M\sigma_{i}^2)\right]
\end{equation}

\begin{multline}
\label{mu1}
\mu_{1}=\sum_{i=1}^{N_1}\left[P_i \dfrac{\tilde{w}_{i}}{\text{sum}(w)}((M+\eta_{i})\sigma_{i}^{2}-\Delta_{i}) 
+(1-P_i)\dfrac{\tilde{w}_{i}}{\text{sum}(w)}((M+\eta_{i})\sigma_{i}^2)\right] \\+ \sum_{i=N_1+1}^{N}\left[\dfrac{w_{i}}{\text{sum}(w)}((M+\eta_{i})\sigma_{i}^2)\right]
\end{multline}

\begin{multline}
\label{sig}
\sigma_{(0)}^2=\sum_{i=1}^{N_1}\left(\frac{\tilde{w}_i}{\text{sum}(w)}\right)^2\left[P_i(1-P_i)\Delta_i^2+2M\sigma_i^4\right]+\sum_{i=N_1+1}^{N}\left(\frac{w_i}{\text{sum}(w)}\right)^2\left[2M\sigma_i^4\right]
\end{multline}

\hrulefill

\end{figure*}

\section{Performance analysis of consensus based detection algorithms}
\label{susceptibility}
In this section, we analyze the effect of data falsification attacks on conventional consensus based detection algorithms. 

First, in Section~\ref{seca}, we characterize the effect of Byzantines on the steady-state performance of the consensus based detection algorithms and determine $\alpha_{blind}$. 
Next, in Section~\ref{secb}, using probability of detection and probability of false alarm as measures of detection performance, we investigate the degradation of transient detection performance of the consensus algorithms with Byzantines.

\subsection{Steady-State Performance Analysis with Byzantines}
\label{seca}

Without loss of generality, we assume that the nodes corresponding to the first $N_1$ indices $i=1,\cdots,N_1$ are Byzantines and the rest corresponding to indices $i=N_1+1,\cdots,N$ are honest nodes. Let us define $w=[\tilde{w}_1,\cdots,\tilde{w}_{N_1},w_{N_1+1}\cdots,w_N]^T$ and  $\text{sum}(w)=\sum_{i=1}^{N_1}\tilde{w}_i+\sum_{i=N_1+1}^{N}w_i$.
\begin{lemma}
\label{La}
For data fusion schemes, the condition to blind the network or to make the deflection coefficient zero is given by
 \begin{center}
$\displaystyle\sum_{i=1}^{N_1} \tilde{w}_i (2P_i \Delta_i-\eta_i\sigma_i^2) =\sum_{i=N_1+1}^N w_i\eta_i\sigma_i^2$.
\end{center}
\end{lemma}

\begin{proof}
The local test statistic $Y_{i}$ has the mean
\[ mean_{i} = \left\{ \begin{array}{rll}
				M\sigma_{i}^{2}  & \mbox{if}\ H_{0} \\
			   (M+\eta_{i})\sigma_{i}^2 & \mbox{if}\ H_{1}\\
				\end{array}\right. 
\]
and the variance

\[ Var_{i} = \left\{ \begin{array}{rll}
				2M\sigma_{i}^{4}  & \mbox{if}\ H_{0} \\
			   2(M+2\eta_{i})\sigma_{i}^{4} & \mbox{if}\ H_{1}.\\
				\end{array}\right. 
\]

The goal of Byzantine nodes is to make the deflection coefficient as small as possible. Since the Deflection Coefficient is always non-negative; the Byzantines seek to make
$\mathcal{D}(\Lambda)=\dfrac{(\mu_{1}-\mu_{0})^{2}}{\sigma_{(0)}^{2}}=0$. The conditional mean $\mu_{k}=\mathbb{E}[\Lambda|H_k]$ and conditional variance $\sigma_{(0)}^2=\mathbb{E}[(\Lambda-\mu_{0})^2|H_0]$ of the global test statistic, $\Lambda=(\sum_{i=1}^{N_1} \tilde{w}_{i}\tilde{Y_{i}}+\sum_{i=N_1+1}^N w_{i}Y_{i})/(\text{sum}(w))$, can be computed and are given by~\eqref{mu0}, \eqref{mu1} and~\eqref{sig}, respectively. 
After substituting values from \eqref{mu0}, \eqref{mu1} and \eqref{sig}, the condition to make $\mathcal{D}(\Lambda)=0$ becomes 

$$\displaystyle\sum_{i=1}^{N_1} \tilde{w}_i (2P_i \Delta_i-\eta_i\sigma_i^2) =\sum_{i=N_1+1}^N w_i\eta_i\sigma_i^2$$
\end{proof}

\begin{figure}[t] 
  \centering
    \includegraphics[width=0.5\textwidth]{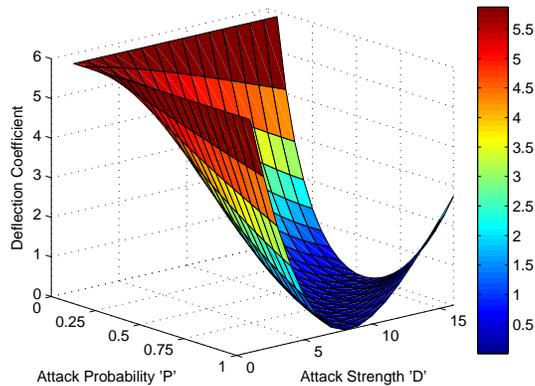}
    \caption{Deflection Coefficient as a function of attack parameters $P$ and $D$.}\label{deflection}
\end{figure}

Note that, when $w_i=\tilde{w}_i=z, \eta_i=\eta,\sigma_i=\sigma,P_i=P,\Delta_i=D,\;\forall i$, the blinding condition simplifies to $\dfrac{N_1}{N}=\dfrac{1}{2}\dfrac{\eta\sigma^2}{PD}$. This condition indicates that by appropriately choosing attack parameters $(P,D)$, an adversary needs less than $50\%$ of sensing data falsifying Byzantines to make the deflection coefficient zero.
 
Next, to gain insights into the solution, we present some numerical results in Figure~\ref{deflection}. We plot the deflection coefficient of global test statistic as a function of attack parameters $P_i=P,\Delta_i=D,\forall i$. 
We consider a $6$-node network with the topology given by the undirected graph shown in Figure~\ref{fig1} to detect a phenomenon.
Nodes $1$ and $2$ are considered to be Byzantines. Channel gains of the nodes are assumed to be $h=[0.8, 0.7, 0.72, 0.61, 0.69, 0.9]$ and weights are given by~\eqref{convweight}.
We also assume that $M=12$, $E_s=5$ and $\sigma_i^2=1$. Notice that, the deflection coefficient is zero when the condition in Lemma~\ref{La} is satisfied. Another observation to make is that the deflection coefficient can be made zero even when only two out of six nodes are Byzantines. Thus,
by appropriately choosing attack parameters $(P,D)$, less than $50\%$ of data falsifying Byzantines are needed to blind the network.

\subsection{Transient Performance Analysis with Byzantines}
\label{secb}
Next, we analyze the detection performance of the data fusion schemes, denoted as $x(t+1)=W^t x(0)$, as a function of consensus iteration $t$ in the presence of Byzantines.
For analytical tractability, we assume that $P_i=P,\;\forall i$. We denote by $w_{ji}^t$ the element of matrix $W^t$ in the $j$th row and $i$th column.
Using these notations, we calculate the probability of detection and the probability of false alarm at the $j$th node at consensus iteration $t$.

For sufficiently large $M$, the distribution of Byzantine's data $\tilde{Y_{i}}$ given $H_{k}$ is a Gaussian mixture which comes from $\mathcal{N}((\mu_{1k})_{i},(\sigma_{1k})_{i}^{2})$ with probability $(1-P)$ and from $\mathcal{N}((\mu_{2k})_{i},(\sigma_{2k})_{i}^{2})$ with probability $P$, where $\mathcal{N}$ denotes the normal distribution and
\begin{equation*}
(\mu_{10})_{i}=M\sigma_{i}^{2},\; (\mu_{20})_{i}=M\sigma_{i}^{2}+\Delta_{i}
\end{equation*}
\begin{equation*}
(\mu_{11})_{i}=(M+\eta_{i})\sigma_{i}^{2},\; (\mu_{21})_{i}=(M+\eta_{i})\sigma_{i}^{2}-\Delta_{i}
\end{equation*}
\begin{equation*}
(\sigma_{10})_{i}^{2}=(\sigma_{20})_{i}^{2}=2M\sigma_{i}^{4},\;\textit{and}\;(\sigma_{11})_{i}^{2}=(\sigma_{21})_{i}^{2}=2(M+\eta_{i})\sigma_{i}^{4}.
\end{equation*}
Now, the probability density function (PDF) of  $x_{ji}^t=w_{ji}^t \tilde{Y_{i}}$ conditioned on $H_{k}$ can be derived as
\begin{equation*}
f(x_{ji}^t|H_{k})=(1-P)\phi(w_{ji}^t(\mu_{1k})_{i},(w_{ji}^t(\sigma_{1k})_{i})^{2})
\end{equation*}
\begin{equation}
\label{weight}
\qquad+P\phi(w_{ji}^t(\mu_{2k})_{i},(w_{ji}^t(\sigma_{2k})_{i})^{2})
\end{equation}
where $\phi(x|\mu,\sigma^{2})$ (for notational convenience denoted as $\phi(\mu,\sigma^{2})$) is the PDF of $X\sim \mathcal{N}(\mu,\sigma^{2})$ and
$\phi(x|\mu,\sigma^{2}) = \frac{1}{{\sigma \sqrt {2\pi } }}e^{{{ - \left( {x - \mu } \right)^2 } \mathord{\left/ {\vphantom {{ - \left( {x - \mu } \right)^2 } {2\sigma ^2 }}} \right. \kern-\nulldelimiterspace} {2\sigma ^2 }}}$.
Next, for clarity of exposition, we first derive our results for a small network with two Byzantine nodes and one honest node. Later we generalize our results for an arbitrary number of nodes, $N$.

Notice that, for the three node case, the transient test statistic $\tilde{\Lambda}_j^t=w_{j1}^t\tilde{Y_1}+w_{j2}^t\tilde{Y_2} + w_{j3}^t Y_3$, is a summation of independent random variables. The conditional PDF of $X_{ji}^t=w_{ji}^t\tilde{Y_i}$ is given in \eqref{weight}. 
Notice that, PDF of $\tilde{\Lambda}_j^t$ is the convolution $(*)$ of $f(x_{j1}^t)=(1-P)\phi(\mu_1^1,(\sigma_1^1)^2)+P\phi(\mu_1^2,(\sigma_1^2)^2)$,  $f(x_{j2}^t)=(1-P)\phi(\mu_2^1,(\sigma_2^1)^2)+P\phi(\mu_2^2,(\sigma_2^2)^2))$ and $f(x_{j3}^t)=\phi(\mu_3^1,(\sigma_3^1)^2)$.
\begin{equation*}
f(z_j^t)=f(x_{j1}^t)*f(x_{j2}^t)*f(x_{j3}^t)
\end{equation*} 
\begin{equation*}
f(z_j^t)=[(1-P)\phi(\mu_1^1,(\sigma_1^1)^2)+P\phi(\mu_1^2,(\sigma_1^2)^2)]*
\end{equation*}
\begin{equation*}
[(1-P)\phi(\mu_2^1,(\sigma_2^1)^2)+P\phi(\mu_2^2,(\sigma_2^2)^2)]*\phi(\mu_3^1,(\sigma_3^1)^2)
\end{equation*}
\begin{equation*}
=(1-P)^2[\phi(\mu_1^1,(\sigma_1^1)^2)*\phi(\mu_2^1,(\sigma_2^1)^2)*\phi(\mu_3^1,(\sigma_3^1)^2)]
\end{equation*}
\begin{equation*}
+(P)^2[\phi(\mu_1^2,(\sigma_1^2)^2)*\phi(\mu_2^2,(\sigma_2^2)^2))*\phi(\mu_3^1,(\sigma_3^1)^2)]
\end{equation*}
\begin{equation*}
+P(1-P)[\phi(\mu_1^2,(\sigma_1^2)^2)*\phi(\mu_2^1,(\sigma_2^1)^2)*\phi(\mu_3^1,(\sigma_3^1)^2)]
\end{equation*}
\begin{equation*}
+(1-P)P[\phi(\mu_1^1,(\sigma_1^1)^2)*\phi(\mu_2^2,(\sigma_2^2)^2)*\phi(\mu_3^1,(\sigma_3^1)^2)]
\end{equation*}
Now, using the fact that convolution of two normal PDFs $\phi(\mu_i,\sigma_i^{2})$ and $\phi(\mu_j,\sigma_j^{2})$ is again normally distributed with mean $(\mu_i+\mu_j)$ and variance $(\sigma_i^{2}+\sigma_j^{2})$, we can derive the results below.
\begin{equation*}
f(z_j^t)=(1-P)^2[\phi(\mu_1^1+\mu_2^1+\mu_3^1,(\sigma_1^1)^2+(\sigma_2^1)^2+(\sigma_3^1)^2)]
\end{equation*}
\begin{equation*}
+P^2[\phi(\mu_1^2+\mu_2^2+\mu_3^1,(\sigma_1^2)^2+(\sigma_2^2)^2+(\sigma_3^1)^2)]
\end{equation*}
\begin{equation*}
+P(1-P)[\phi(\mu_1^2+\mu_2^1+\mu_3^1,(\sigma_1^2)^2+(\sigma_2^1)^2+(\sigma_3^1)^2)]
\end{equation*}
\begin{equation*}
+(1-P)P[\phi(\mu_1^1+\mu_2^2+\mu_3^1,(\sigma_1^1)^2+(\sigma_2^2)^2+(\sigma_3^1)^2)].
\end{equation*}
Due to the probabilistic nature of the Byzantine's behavior, it may behave as an honest node with a probability $(1-P_i)$. Let $S$ denote the set of all combinations of such Byzantine strategies:
\begin{equation}
\label{eqI}
 S=\{\{b_{1},b_{2}\},\{h_{1},b_{2}\},\{b_{1},h_{2}\},\{h_{1},h_{2}\}\}
\end{equation}
where by $b_i$ we mean that Byzantine node $i$ behaves as a Byzantine and by $h_i$ we mean that Byzantine node $i$  behaves as an honest node. Let $A_s \in U$ denote the indices of honest nodes in the strategy combination $s$, then, from~\eqref{eqI} we have
\begin{equation*}
 U=\{A_{1}=\{\},A_{2}=\{1\},A_{3}=\{2\},A_{4}=\{1,2\}\}
\end{equation*} 
\begin{equation*}
 U^c=\{A_{1}^c=\{1,2\},A_{2}^c=\{2\},A_{3}^c=\{1\},A_{4}^c=\{\}\}
\end{equation*}
where $\{\}$ is used to denote the null set and $m_s$ to denote the cardinality of subset $A_{s}\in U$. Using these notations, we generalize our results for any arbitrary $N$.  
\begin{lemma}
The test statistic of node $j$ at consensus iteration $t$, i.e., $\tilde{\Lambda}_j^t=\sum_{i=1}^{N_1}w_{ji}^t\tilde{Y_{i}}+\sum_{i=N_1+1}^{N}w_{ji}^t Y_{i}$ is a Gaussian mixture with PDF
\begin{eqnarray*}
f(\tilde{\Lambda}_j^t|H_{k})&=& \sum_{A_{s}\in U}P^{N_1-m_s}(1-P)^{m_s}\phi\left((\mu_{k})_{A_{s}}+\sum_{i=N_1+1}^{N}w_{ji}^t(\mu_{1k})_{i},\sum_{i=1}^{N}(w_{ji}^t(\sigma_{1k})_{i})^{2})\right)
\end{eqnarray*}
\begin{center}
with $(\mu_{k})_{A_{s}}=\underset{u\in A_{s}}{\sum} w_{ju}^t(\mu_{1k})_{j}+\underset{u\in A_{s}^{c}}{\sum} w_{ju}^t(\mu_{2k})_{j}$.
\end{center}
\end{lemma}

The performance of the detection scheme in the presence of Byzantines can be represented in terms of the probability of detection and the probability of false alarm of the network. 

\begin{proposition}
\label{propo1}
The probability of detection and the probability of false alarm of node $j$ at consensus iteration $t$ in the presence of Byzantines can be represented as
\begin{footnotesize}
\begin{equation*}
P_{d}^t(j)=\sum_{A_{s}\in U}P^{N_1-m_s}(1-P)^{m_s} Q\left(\frac{\lambda-(\mu_{1})_{A_{s}}-\sum_{i=N_1+1}^{N}w_{ji}^t(\mu_{11})_{i}}{\sqrt{\sum_{i=1}^{N}(w_{ji}^t(\sigma_{11})_{i})^{2})}}\right),
\end{equation*}
\begin{equation*}
P_{f}^t(j)=\sum_{A_{s}\in U}P^{N_1-m_s}(1-P)^{m_s} Q\left(\frac{\lambda-(\mu_{0})_{A_{s}}-\sum_{i=N_1+1}^{N}w_{ji}^t(\mu_{10})_{i}}{\sqrt{\sum_{i=1}^{N}(w_{ji}^t(\sigma_{10})_{i})^{2})}}\right).
\end{equation*}
\end{footnotesize}
\end{proposition}

\newtheorem{remark}{Remark} 
\begin{remark}
Notice that, the expressions of probability of detection $P_{d}^t(j)$ and probability of false alarm $P_{f}^t(j)$ for the $N_1$ Byzantine node case involves $2^{N_1}$ combinations (cardinality of $U$ is $2^{N_1}$). It, however, can be represented compactly by vectorizing the expressions, i.e., 
\begin{equation*}
P_d^t(j)=\mathbf{1^T}\left(\mathbf{b}\otimes Q\left(\frac{\mbox{$\lambda \mathbf{1}$-\boldmath{$\mu_1$}}}{\sqrt{\sum_{i=1}^{N}(w_{i}(\sigma_{10})_{i})^{2})}}\right) \right)
\end{equation*} 
with $\mbox{\boldmath$\mu_1$}=[A\mbox{\boldmath$w_j^t\mu_{11}$}+A^c\mbox{\boldmath$w_j^t\mu_{21}$}-\sum\limits_{i=N_1+1}^{N}w_{ji}^t(\mu_{11})_{i}]$, $B=(1-P)A+P A^c$ and $\mathbf{b}=[\mathbf{B^1}\otimes\cdots\otimes\mathbf{B^{N_1}}]$,
where boldface letters represent vectors, $\otimes$ symbol represents element-wise multiplication, $Q(\cdot)$ represents element wise Q function operation, i.e., $Q(x_1,\cdots,x_{N_1})=[Q(x_1),\cdots,Q(x_{N_1})]^T$, $\mathbf{B^i}$ is $i$th column of matrix $B$, $\mbox{\boldmath$w_j^t\mu_{u1}$}=[w_{j1}^t \mu_{u1}^1,\cdots,w_{jN_1}^t \mu_{u1}^{N_1}]^T$, matrix $A_{(2^{N_1}*{N_1})}$ is the binary representation of decimal numbers from $0$ to ${N_1}-1$ and $A^c$ is the matrix after interchanging $1$ and $0$ in matrix $A$.

Similarly, the expression for the probability of false alarm $P_{f}^t(j)$ can be vectorized into a compact form.
\end{remark}

Next, to gain insights into the results given in Proposition~\ref{propo1}, we present some numerical results in Figures~\ref{Pd} and \ref{Pf}.
We consider the $6$-node network shown in Figure~\ref{fig1} where the nodes employ the consensus algorithm~\ref{consensusalgo} with $\epsilon=0.6897$ to detect a phenomenon.
Nodes $1$ and $2$ are considered to be Byzantines.
We also assume that $\eta_i=10$, $\sigma_i^2=2$,  $\lambda=33$ and $w_i=1$. Attack parameters are assumed to be $(P_i,\Delta_i)=(0.5,6)$ and $\tilde{w}_i=1.1$.
To characterize the transient performance
of the weighted average consensus algorithm, in Figure~\ref{pd}, we plot the probability of detection as a function of the number of consensus iterations when Byzantines are not falsifying their data, i.e., $(\Delta_i=0,\tilde{w}_i=w_i)$. Next, in Figure~\ref{pdb}, we plot the probability of detection as a function of the number of consensus iterations in the presence of Byzantines. It can be seen that the detection performance degrades in the presence of Byzantines.
In Figure~\ref{pf}, we plot the probability of false alarm as a function of
the number of consensus iterations when Byzantines are not falsifying their data, i.e., $(\Delta_i=0,\tilde{w}_i=w_i)$.
Next, in Figure~\ref{pfb}, we plot the probability of false alarm as a function of the number of consensus iterations in the presence of Byzantines.
From both Figures~\ref{Pd} and \ref{Pf}, it can be seen  that the Byzantine attack can severely degrade transient detection performance. 

\begin{figure*}[t]
\centering
\subfigure[]{
\includegraphics[height=0.2\textheight, width=0.42\textwidth,clip=true]{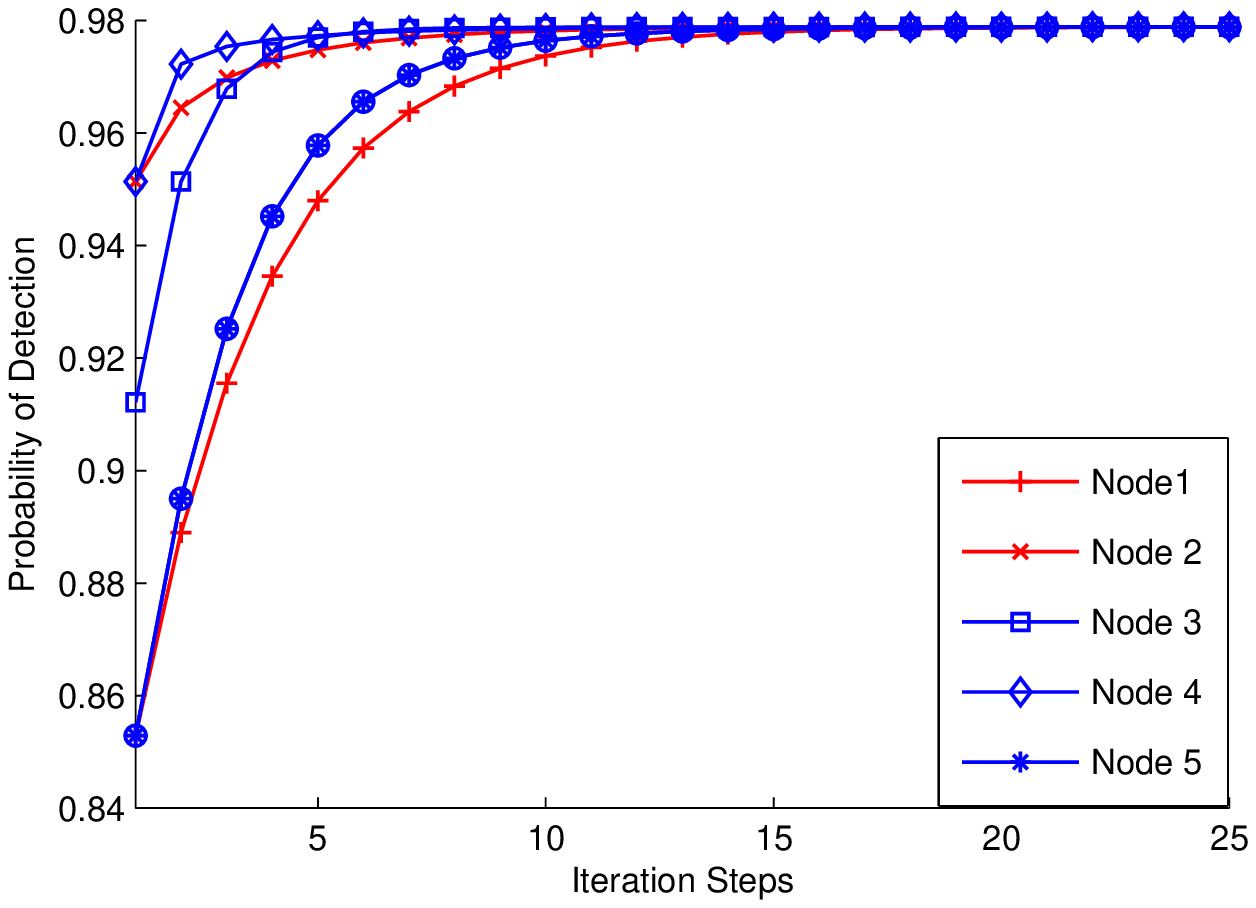}
\label{pd}}
\hspace{0.1in}
\subfigure[] {
\includegraphics[height=0.2\textheight, width=0.42\textwidth,clip=true]{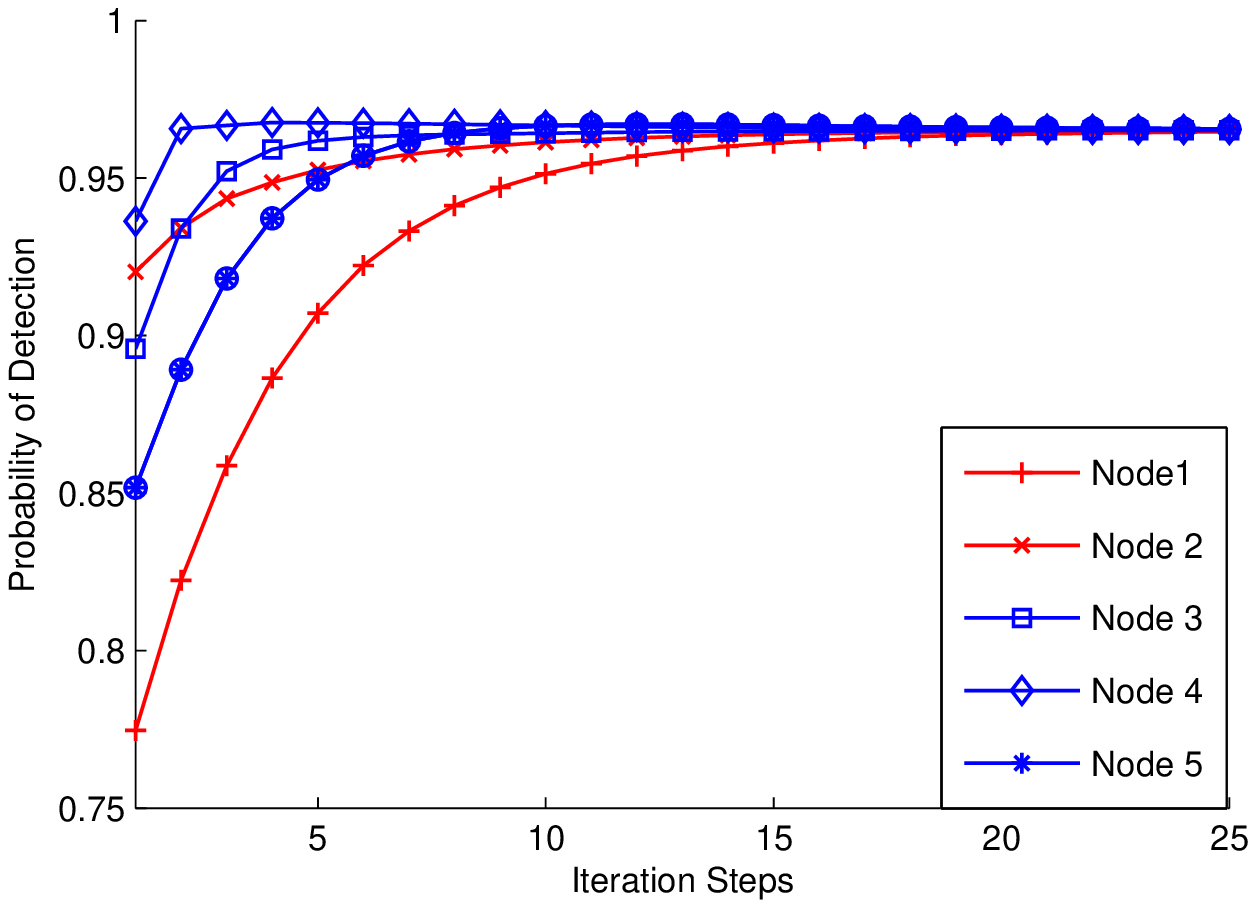}
\label{pdb} }
\caption{\subref{pd} Probability of detection as a function of consensus iteration steps.  \subref{pdb} Probability of detection as a function of consensus iteration steps with Byzantines.}
\label{Pd}
%\vspace{-0.1in}
\end{figure*}
\begin{figure*}[t]
\centering
\subfigure[]{
\includegraphics[height=0.2\textheight, width=0.42\textwidth,clip=true]{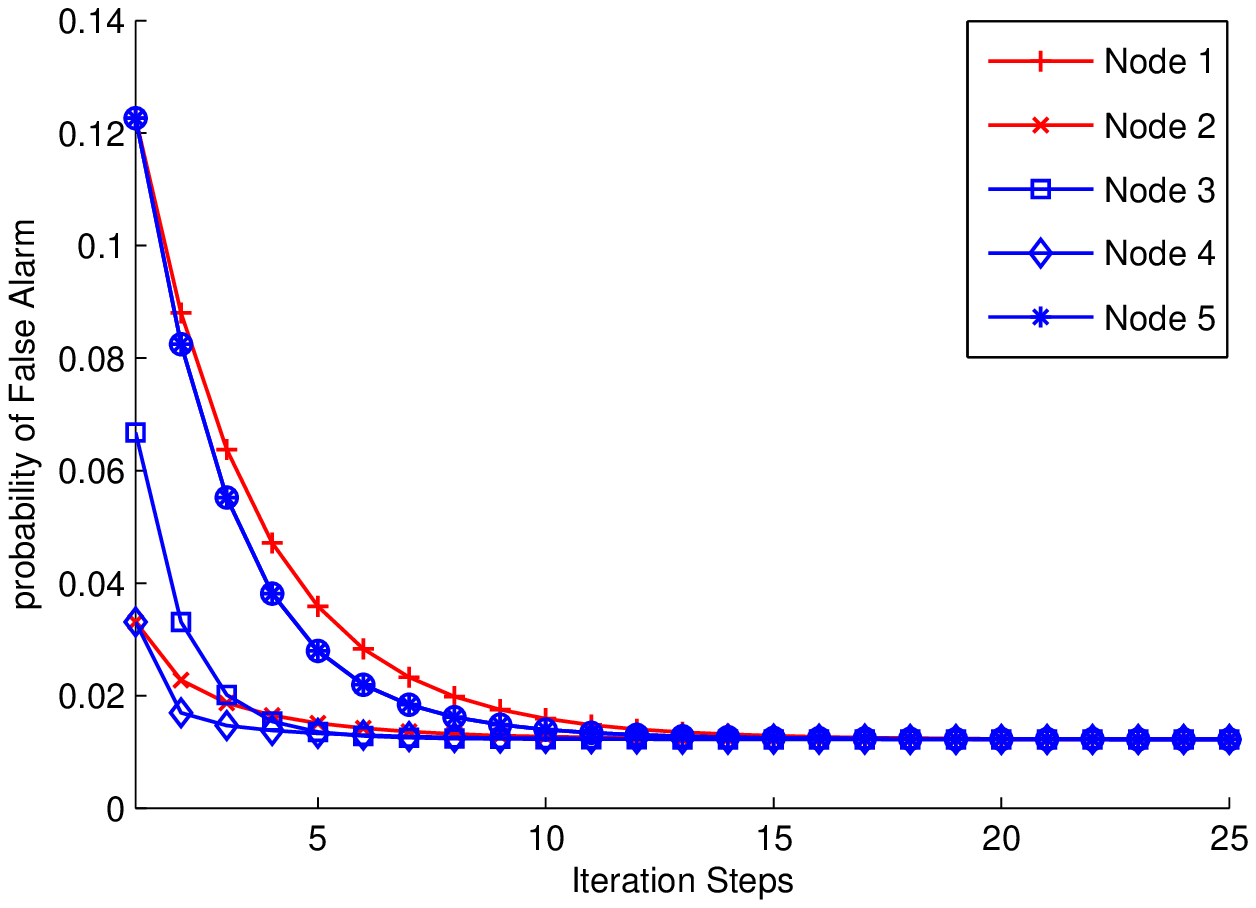}
\label{pf}}
\hspace{0.1in}
\subfigure[] {
\includegraphics[height=0.2\textheight, width=0.42\textwidth,clip=true]{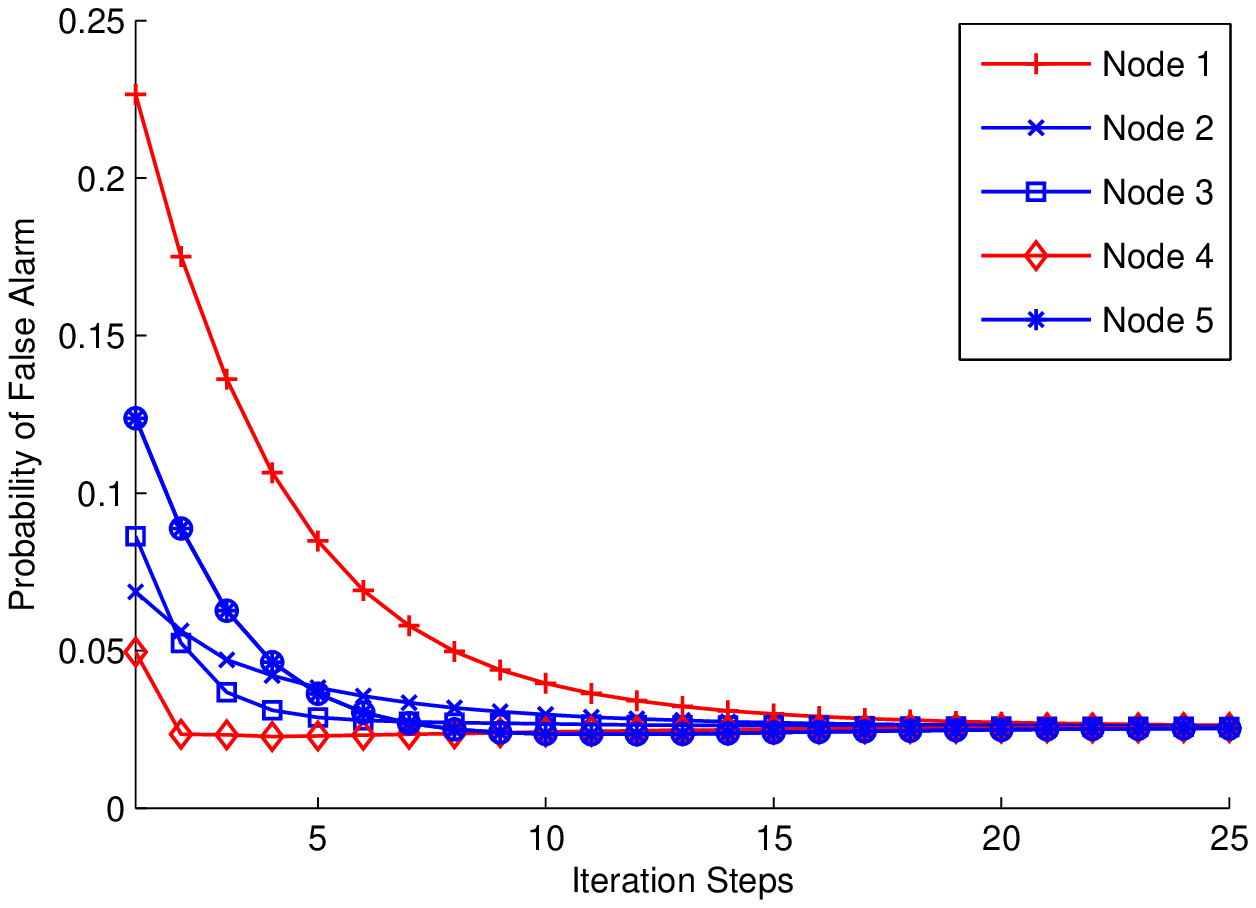}
\label{pfb} }
\caption{\subref{pf} Probability of false alarm as a function of consensus iteration steps.  \subref{pfb} Probability of false alarm as a function of consensus iteration steps with Byzantines.}
\label{Pf}
%\vspace{-0.1in}
\end{figure*}

From the discussion in this section, we can see that Byzantines can severely degrade both the steady-state and the transient detection performance of conventional consensus based detection algorithms. As mentioned earlier, a data falsifying Byzantine $i$ can tamper its weight $w_i$ as well as its sensing data $Y_i$ to degrade detection performance. One approach to mitigate the effect of sensing data falsification is to assign weights based on the quality of the data. In other words, lower weight is assigned to the data of the node identified as a Byzantine. However, to implement this approach one has to address the following two issues.

First, in the conventional weighted average consensus algorithm, weight $w_{i}$ given to node $i$'s data is assigned by the node itself. 
Thus, a Byzantine node can always set a higher weight to its manipulated information and the final statistics will be dominated by the Byzantine nodes' local statistic that will lead to degraded detection performance. 
It will be impossible for any algorithm to detect this type of malicious behavior, since any weight
that a Byzantine chooses for itself is a legitimate value that could also have been chosen by a node that is functioning correctly. Thus, the conventional consensus algorithm cannot be used in the presence of an attacker. 

Second, as will be seen later, the optimal weights assigned to nodes' sensing data depend on the following unknown parameters: identity of the nodes (i.e., honest or Byzantine) and underlying statistical distribution of the nodes' data. 

In the next section, we address these concerns by proposing a learning based robust weighted average consensus algorithm. 

\section{A Robust Consensus Based Detection Algorithm}
\label{protection}
In order to address the first issue, we propose a consensus algorithm in which the weight for node $i$'s information is assigned (or updated) by neighbors of the node $i$ rather than by node $i$ itself. 
Note that, networks deploying such an algorithm is more robust to weight manipulation because if a Byzantine node $j$ wants to lower the weight assigned to the data of its honest neighbor $i$ in the global test statistic, it has to make sure that a majority of the neighbors of $i$ put the same lower weight as $j$. In other words, every honest node should have majority of its neighbors that are Byzantines, otherwise, it can be treated as a consensus disruption attack and Byzantines can be easily identified detected by techniques such as those given in~\cite{sundar,paper5}.

\subsection{Distributed Algorithm for Weighted Average Consensus} 
\label{consensus}
In this section, we address the following questions: does there exist a distributed algorithm that solves the weighted average consensus problem while satisfying the condition that weights must be assigned or updated by neighbors $\mathcal{N}_{i}$ of the node $i$ rather than by the node $i$ itself? If it exists, then, under what conditions or constraints does the algorithm converge?  

We consider a network with $N$ nodes with a fixed  and connected topology $G(V,E)$. Next, we state Perron-Frobenius theorem~\cite{horn}, which will be used later for the design and analysis of our robust weighted average consensus algorithm.
\begin{theorem}[\cite{horn}]
Let $W$ be a primitive nonnegative matrix with left and right eigenvectors $u$ and $v$, respectively, satisfying $Wv=v$ and $u^TW=u^T$. Then, $\lim_{k\to\infty}W^k=\frac{vu^T}{v^Tu}$.
\end{theorem}

Using the above theorem, we take a reverse-engineering approach to design a modified Perron matrix $\hat{W}$ which has the weight vector $w=[w_{1},w_{2},\cdots,w_{N}]^{T}$, $w_{i}>0$,  $\forall i$ as its left eigenvector and $\vec{1}$ as its right eigenvector corresponding to eigenvalue $1$. From the above theorem, if the modified Perron matrix $\hat{W}$ is primitive and nonnegative, then, a weighted average consensus can be achieved. 
Now, the problem boils down to designing such a $\hat{W}$ which meets our requirement that weights are assigned or updated by the neighbors $\mathcal{N}_{i}$ of node $i$ rather than by node $i$ itself. 

Next, we propose a modified Perron matrix ${\hat{W}}=I-\epsilon (T\otimes L)$ where $L$ is the original graph Laplacian, $\otimes$ is element-wise matrix multiplication operator, and $T$ is a transformation given by
\[ [T]_{ij} = \left\{ \begin{array}{rll}
			
				\frac{{\underset{j \in \mathcal{N}_{i}}{\sum}} w_{j}}{l_{ii}}  & \mbox{if} \; {i=j}\\
				
				w_{j}\;\;\;\;\;\;  &\text{otherwise}.
				\end{array}\right.  
\]
Observe that, the above transformation $T$ satisfies the condition that weights are assigned or updated by neighbors $\mathcal{N}_{i}$ of node $i$ rather than by node $i$ itself. 
Based on the above transformation $T$, we propose our  distributed consensus algorithm:
\begin{center}
$x_{i}(k+1)=x_{i}(k)+{\epsilon}\underset{j\in \mathcal{N}_{i}}{\sum}w_{j}(x_{j}(k)-x_{i}(k))$. 
\end{center}
Let us denote the modified Perron matrix by $\hat{W}=I-\epsilon \hat{L}$. 

Next, we explore the properties of the modified Perron matrix $\hat{W}$ and show that it satisfies the requirements of the Perron-Frobenius theorem~\cite{horn}. These properties will later be utilized to prove the convergence of our proposed consensus algorithm.

\begin{lemma}
Let $G$ be a connected graph with $N$ nodes. Then, the modified Perron matrix ${\hat{W}}=I-\epsilon (T\otimes L)$, with $0<\epsilon<\frac{1}{\displaystyle\max_{i}({\underset{j \in \mathcal{N}_{i}}{\sum}} w_{j})}$ satisfies the following properties.
\begin{enumerate}
\item $\hat{W}$ is a nonnegative matrix with left eigenvector $w$ and right eigenvector $\vec{1}$ corresponding to eigenvalue $1$;
\item All eigenvalues of $\hat{W}$ are in a unit circle;
\item $\hat{W}$ is a primitive matrix\footnote{A matrix is primitive if it is non-negative and its $m$th power is positive for some natural number $m$.}.
\end{enumerate}
\end{lemma}
 \begin{proof}
Notice that, $\hat{W}\vec{1}=\vec{1}-{\epsilon}(T\otimes L)\vec{1}=\vec{1}$ and $w^T\hat{W}= w^T-{\epsilon}w^T(T\otimes L) =w^T$. This implies that $\hat{W}$ has left eigenvector $w$ and right eigenvector $\vec{1}$ corresponding to eigenvalue $1$. 
To show that $\hat{W}=I+{\epsilon}T\otimes A-{\epsilon}T\otimes D$ is non-negative, it is sufficient to show that: $w>0$, $\epsilon>0$ and $\epsilon({\max_{i}({\underset{j \in \mathcal{N}_{i}}{\sum}} w_{j})})\leq1, \forall i$. Since $w$ is the left eigenvector of $\hat{L}$ and $w>0$, $\hat{W}$ is non-negative if and only if
\begin{center}
$0<\epsilon\leq\frac{1}{\max_{i}({\underset{j \in \mathcal{N}_{i}}{\sum}} w_{j})}.$ 
\end{center} 

To prove part $2)$, notice that all the eigenvectors of $\hat{W}$ and $\hat{L}$ are the same. Let $\gamma_j$ be
the $j$th eigenvalue of $\hat{L}$, then, the $j$th eigenvalue of $\hat{W}$ is $\lambda_j=1-\epsilon \gamma_j$. 
Now, part $2)$ can be proved by applying Gershgorin theorem~\cite{horn} to the modified Laplacian matrix $\hat{L}$.

To prove part $3)$, note that $G$ is strongly connected and,
therefore, $\hat{W}$ is an irreducible matrix~\cite{horn}. Thus, to prove that $\hat{W}$ is a primitive matrix, it is sufficient\footnote{An irreducible stochastic matrix is primitive if it has only one eigenvalue with maximum modulus.} to show that $\hat{W}$ has a single eigenvalue
with maximum modulus of $1$.
In~\cite{paper2}, the authors showed that when $0<\epsilon<\displaystyle\max_{i}(\underset{j\neq i}\sum{a_{ij}})$, the original Perron matrix $W$ has only one eigenvalue with maximum modulus $1$ at its spectral radius. Using a similar logic, $\hat{W}$ is a primitive matrix if
\begin{center}
$0<\epsilon<\frac{1}{\displaystyle\max_{i}({\underset{j \in \mathcal{N}_{i}}{\sum}} w_{j})}.$   
\end{center}
\end{proof}

\begin{theorem}
\label{th1}
Consider a network with fixed  and strongly connected undirected topology $G(V,E)$ that employs the distributed
consensus algorithm  
\begin{center}
$x_{i}(k+1)=x_{i}(k)+{\epsilon}\underset{j\in \mathcal{N}_{i}}{\sum}w_{j}(x_{j}(k)-x_{i}(k))$ 
\end{center} 
where 
\begin{center}
$0<\epsilon<\frac{1}{\displaystyle\max_{i}({\underset{j \in \mathcal{N}_{i}}{\sum}} w_{j})}.$
\end{center} 
Then, consensus with $x_i^{*}=\frac{\sum_{i=1}^N w_{i}x_i(0)}{\sum_{i=1}^n w_{i}},\forall i$ is reached asymptotically.
\end{theorem}
\begin{proof}
A consensus is
reached asymptotically, if the limit $\displaystyle\lim_{k\rightarrow\infty}\hat{W}^k$ exists. According to Perron-Frobenius theorem~\cite{horn}, this limit exists for primitive matrices. Note that, $\vec{1}=[1,\cdots,1]^T$ and $w$ are right and left eigenvectors of the primitive
nonnegative matrix $\hat{W}$ respectively. Thus, from~\cite{horn}
\begin{eqnarray*}
&&
\lim_{k\to\infty}x(k)= \lim_{k\to\infty}(\hat{W})^{k}x(0)\\
&&
x^{*}= \vec{1}\frac{w^{T}x(0)}{w^{T}\vec{1}}\\
&&
x^{*}=\vec{1}\frac{\sum_{i=1}^N w_{i}x_{i}(0)}{\sum_{i=1}^n w_{i}}
\end{eqnarray*}
\end{proof}

\begin{figure}[t]
  \centering
    \includegraphics[width=0.5\textwidth]{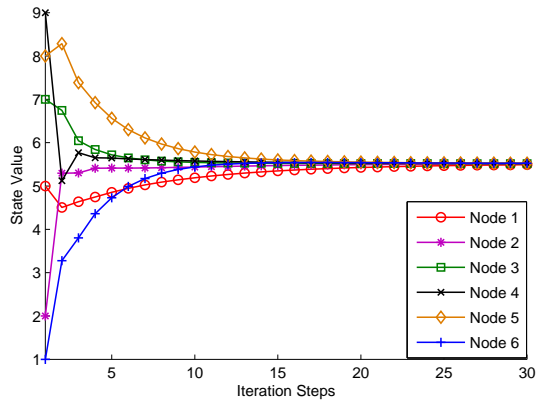}
    \caption{ Convergence of the network with a 6-nodes  
    ($\epsilon=0.3$).}\label{conv} 
\end{figure}

Next, to gain insights into the convergence property of the proposed algorithm, we present some
numerical results in Figure~\ref{conv}.
We consider the $6$-node network shown in Figure~\ref{fig1} where the nodes employ the proposed algorithm (with $\epsilon=0.3$) to reach a consensus.  
Next, we plot the updated state values at each node as a function of consensus iterations.
We assume that the initial data vector is $x(0)=[5,2,7,9,8,1]^T$ and the weight vector is $w=[0.65,0.55,0.48,0.95,0.93,0.90]^T$.
Note that, the parameter $\epsilon$ satisfies the condition mentioned in Theorem~\ref{th1}.
Figure~\ref{conv} shows the convergence of the proposed algorithm iterations for a fixed communication graph. It is observed that within $20$ iterations consensus has been
reached on the global decision statistics or weighted average of the initial values (states).

In the proposed consensus algorithm, weights assigned to node $i$'s data are updated by neighbors of the node $i$ rather than by node $i$ itself which addresses the first issue.

\subsection{Adaptive Design of the Update Rules based on Learning of Nodes' Behavior}

Next, to address the second issue, we exploit the statistical distribution of the sensing data and devise techniques to mitigate the influence of Byzantines on the distributed detection system. We propose a three-tier mitigation scheme where the following three steps are performed at each node: $1$) identification of Byzantine neighbors, $2$) estimation of parameters of identified Byzantine neighbors, and $3$) adaptation of consensus algorithm (or update weights) using estimated parameters. 

We first present the design of distributed optimal weights for the honest/Byzantine nodes assuming that the identities of the nodes are known. Later we will explain how the identity of nodes (i.e., honest/Byzantine) can be determined. 

\subsubsection{Design of Distributed Optimal Weights in the Presence of Byzantines}

In this subsection, we derive closed form expressions for the distributed optimal
weights which maximize the deflection coefficient. First, we consider the global test
statistic $\Lambda=\frac{\sum_{i=1}^{N_1} w_{i}^B\tilde{Y_{i}}+\sum_{i=N_1+1}^N w_{i}^HY_{i}}{\sum_{i=1}^{N_1} w_{i}^B+\sum_{i=N_1+1}^N w_{i}^H}$ and obtain a closed form solution for optimal centralized
weights. Then, we extend our analysis to the distributed
scenario.
Let us denote by $\delta_i^B$, the centralized weight given to the Byzantine node and by $\delta_i^H$, the centralized weight given to the Honest node. 
By considering $\delta_i^B=w_i^B/\left(\sum\limits_{i=1}^{N_1}w_i^B+\sum\limits_{i=N_1+1}^{N}w_i^H\right)$ and $\delta_i^H=w_i^H/\left(\sum\limits_{i=1}^{N_1}w_i^B+\sum\limits_{i=N_1+1}^{N}w_i^H\right)$, the optimal weight design problem can be stated formally as: 

\begin{equation*}
\begin{aligned}
& \underset{\{\delta_{i}^B\}_{i=1}^{N_1},\{\delta_{i}^H\}_{i=N_1+1}^{N}}{\max}
\dfrac{(\mu_{1}-\mu_{0})^{2}}{\sigma_{(0)}^{2}}  \\
&\qquad \qquad\text{st.}\quad \sum\limits_{i=1}^{N_1}\delta_i^B+\sum\limits_{i=N_1+1}^{N}\delta_i^H=1
\end{aligned}
\end{equation*}

where $\mu_{1}$, $\mu_{0}$ and $\sigma_{(0)}^{2}$ are given as in \eqref{mu0}, \eqref{mu1} and \eqref{sig}, respectively.
The solution of the above problem is presented in the next lemma.
\begin{lemma}
Optimal centralized weights which maximize the deflection coefficient are given as
\begin{eqnarray*}
\delta_i^B &=&\frac{w_i^B}{\sum\limits_{i=1}^{N_1}w_i^B+\sum\limits_{i=N_1+1}^{N}w_i^H},\\
\delta_i^H&=&\frac{w_i^H}{\sum\limits_{i=1}^{N_1}w_i^B+\sum\limits_{i=N_1+1}^{N}w_i^H}
\end{eqnarray*}
where
$w_i^B =\dfrac{(\eta_i\sigma_i^2-2P_i\Delta_i)}{\Delta_i^2 P_i (1-P_i) + 2M\sigma_i^4}$ and
$w_i^H=\dfrac{\eta_i}{2M\sigma_i^2}$.
\end{lemma}
\begin{proof}
The above results can be obtained by equating the derivative of the deflection coefficient to zero.
\end{proof}

\begin{figure}[t]
  \centering
    \includegraphics[width=0.5\textwidth]{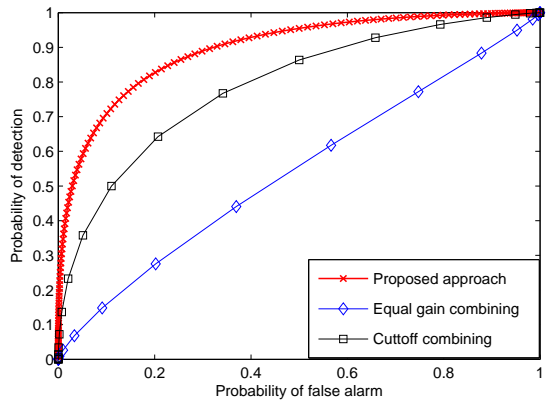}
    \caption{ROC for different protection approaches}\label{rocopt} 
\end{figure}
\begin{remark}
\label{rem2}

Distributed optimal weights can be chosen as $w_i^B$ and $w_i^H$. Thus, the value of the global test statistic (or final weighted average consensus) is the same as the optimal centralized weighted combining\footnote{Note that, weights $w_i^B$ can be negative and in that case convergence of the proposed algorithm is not guaranteed. However, this situation can be dealt off-line by adding a constant value to make $w_i^B\geq 0$ and changing the threshold $\lambda$ accordingly. More specifically, by choosing a constant $c$ such that $\left(w_i^B+\frac{c}{x_i(0)}\right)\geq 0,\forall i$ and $\lambda\leftarrow\lambda+\beta c$ where $\beta$ is number of nodes with $w_i^B<0$.}.
\end{remark}

Next, to gain insights into the solution, we present some numerical results in Figure~\ref{rocopt} that
corroborate our theoretical results. We assume that $M=12$,
$\eta_i=3$, $\sigma_i^2=0.5$ and the attack parameters are $(P_i,\Delta_i)=(0.5,9)$. In Figure~\ref{rocopt}, we compare our proposed weighted average consensus based detection scheme with the equal gain combining scheme\footnote{In equal gain combining scheme, all the nodes (including Byzantines) are assigned the same weight.} and the scheme where Byzantines are cut off or removed from the fusion process.

It can be clearly seen from the figure that our proposed scheme performs better than the rest of the schemes.

Notice that, the optimal weights for the Byzantines are functions of the attack parameters $(P_i,\;\Delta_i)$, which may not be known to the neighboring nodes in practice. In addition, the parameters of the honest nodes might also not be known. Thus, we propose a technique to learn or estimate these parameters. We then use these estimates to adaptively design the local fusion rule which are updated after each learning iteration.

\subsubsection{Identification, Estimation, and Adaptive Fusion Rule}

The first step at each node $m$ is to determine the identity ($I^i\in\{H,B\}$) of its neighboring nodes $i\in \mathcal{N}_m$. 
Notice that, if node $i$ is an honest node, its data under hypothesis $H_k$ is normally distributed $\mathcal{N}((\mu_{1k})_{i},(\sigma_{1k})_{i}^{2})$.
On the other hand, if node $i$ is a Byzantine node, its data under hypothesis $H_k$ is a Gaussian mixture which comes from $\mathcal{N}((\mu_{1k})_{i},(\sigma_{1k})_{i}^{2})$ with probability $(\alpha_1^i=1-P_i)$ and from $\mathcal{N}((\mu_{2k})_{i},(\sigma_{2k})_{i}^{2})$ with probability $\alpha_2^i=P_i$. 
Therefore, determining the identity ($I^i\in\{H,B\}$) of neighboring nodes $i\in \mathcal{N}_m$ can be posed as a hypothesis testing problem:

\begin{center}
\vspace{-0.5in}
\begin{align*}
& I_0\; (I^i=H):\; Y_i\; \text{is generated from a Gaussian distribution under each hypothesis $H_k$;}\\
& I_1\; (I^i=B):\; Y_i\; \text{is generated from a Gaussian mixture distribution under each hypothesis $H_k$.}
\end{align*}
\end{center}

Node classification can then be achieved using the maximum likelihood decision rule:
\begin{footnotesize}
\begin{equation}
f(Y_i|\;I_0) \quad \mathop{\stackrel{H}{\gtrless}}_{B} \quad  f(Y_i|\;I_1)
\end{equation}
\end{footnotesize}
where $f(\cdot|\;I_l)$ is the probability density function (PDF) under each hypothesis $I_l$.
However, the parameters of the distribution are not known. Next, we propose a technique to learn these parameters. For an honest node $i$, the parameters to be estimated are $((\mu_{1k})_{i},(\sigma_{1k})_{i}^{2})$ and for Byzantines the unknown parameter set to be estimated is $\theta=\{\alpha_j^i,(\mu_{jk})_{i},(\sigma_{jk})_{i}^{2}\}$, where $k=\{0,1\}$, $j=\{1,2\}$ and $i=1,\cdots,N_m$, for  $N_m$ neighbor nodes. 
These parameters are estimated by observing the data over
multiple learning iterations. In each iteration $t$, every node in the network observes the data coming from their neighbors for $D$ detection intervals to learn their respective parameters. It is assumed that each node has the knowledge of the true hypothesis for $D$ detection intervals (or history) through a feedback mechanism. 

First, we explain how the unknown parameter set for the distribution under the null hypothesis $(I_0)$ can be estimated. Let us denote the data coming from an honest neighbor node $i$ as $\mathbf{Y}_i(t)=[{y}_i^0(1),\cdots,{y}_i^0(D_1(t)),{y}_i^1(D_1(t)+1),\cdots,{y}_i^1(D)]$ where $D_1(t)$ denotes the number of times $H_0$ occurred in learning iteration $t$ and ${y}_i^k$ denotes the data of node $i$ when the true hypothesis was $H_k$. To estimate the parameter set, $((\mu_{1k})_{i},(\sigma_{1k})_{i}^{2})$, of an honest neighboring node, one can employ a maximum likelihood based estimator (MLE). We use $((\hat{\mu}_{1k})_{i} (t),(\hat{\sigma}_{1k})_{i}^{2}(t))$ to denote the estimates at learning iteration $t$, where each learning iteration consists of $D$ detection intervals.
The ML estimate of $((\mu_{1k})_{i},(\sigma_{1k})_{i}^{2})$ can be written in a recursive form as following:

\begin{figure*}[t!]
\begin{small}
\setcounter{mytempeqncnt}{\value{equation}}
\begin{equation}
\label{sige0}
(\hat{\sigma}_{10})_{i}^{2}(t+1)=
\frac{(\sum\limits_{r=1}^{t}D_1(r))[(\hat{\sigma}_{10})_{i}^{2}(t)+((\hat{\mu}_{10})_{i}(t+1)-(\hat{\mu}_{10})_{i}(t))^2]+\sum\limits_{d=1}^{D_1(t+1)}[y_i^0(d)-(\hat{\mu}_{10})_{i}(t+1)]^2}{\sum\limits_{r=1}^{t+1}D_1(r)} 
\end{equation}

\begin{equation}
\label{sige1}
(\hat{\sigma}_{11})_{i}^{2}(t+1)= 
\frac{\sum\limits_{r=1}^{t}(D-D_1(r))[(\hat{\sigma}_{11})_{i}^{2}(t)+((\hat{\mu}_{11})_{i}(t+1)-(\hat{\mu}_{11})_{i}(t))^2]+\sum\limits_{d=1}^{D-D_1(t+1)}[y_i^1(d)-(\hat{\mu}_{11})_{i}(t+1)]^2}{\sum\limits_{r=1}^{t+1}(D-D_1(r))}
\end{equation}
\hrulefill
\end{small}
\end{figure*}
 
\begin{eqnarray}
(\hat{\mu}_{10})_{i}(t+1)&=& \frac{\sum\limits_{r=1}^{t}D_1(r)}{\sum\limits_{r=1}^{t+1}D_1(r)} (\hat{\mu}_{10})_{i}(t)+ \frac{1}{\sum\limits_{r=1}^{t+1}D_1(r)}\sum\limits_{d=1}^{D_1(t+1)} y_i^0(d)\\
(\hat{\mu}_{11})_{i}(t+1)&=&\frac{\sum\limits_{r=1}^{t}(D-D_1(r))}{\sum\limits_{r=1}^{t+1}(D-D_1(r))} (\hat{\mu}_{11})_{i}(t)+ \frac{1}{\sum\limits_{r=1}^{t+1}(D-D_1(r))}\sum\limits_{d=D_1(t+1)}^{D} y_i^1(d)
\end{eqnarray}

where expressions for $(\hat{\sigma_{10}})_{i}^{2}$ and $(\hat{\sigma_{11}})_{i}^{2}$ are given in~\eqref{sige0} and ~\eqref{sige1}. Observe that, by writing these expressions in a recursive manner, we need to store only  $D$ data samples at any given learning iteration $t$, but effectively use all $tD$ data samples to determine the estimates. 

Next, we explain how the unknown parameter set for the distribution under the alternate hypothesis $(I_1)$ can be estimated. Since the data is distributed as a Gaussian mixture, we employ the expectation-maximization (EM) algorithm to estimate the unknown parameter set for Byzantines. Let us denote the data coming from a Byzantine neighbor $i$ as $\mathbf{\tilde{Y}}_i(t)=[{\tilde{y}}_i^0(1),\cdots,{\tilde{y}}_i^0(D_1(t)),{\tilde{y}}_i^1(D_1(t)+1),\cdots,{\tilde{y}}_i^1(D)]$ where $D_1(t)$ denotes the number of times $H_0$ occurred in learning iteration $t$ and ${\tilde{y}}_i^k$ denotes the data of node $i$ when the true hypothesis was $H_k$. Let us denote the hidden variable as $z_j$ with $j=\{1,2\}$ or ($Z=[z_1,z_2]$). Now, the joint conditional PDF of $\tilde{y}_i^k$ and $z_j$, given the parameter set, can be calculated to be 

\begin{eqnarray*}
P(\tilde{y}_i^k(d),z_j|\theta)&=& P(z_j|\tilde{y}_i^k(d),\theta)P(\tilde{y}_i^k(d)|(\mu_{jk})_{i},(\sigma_{jk})_{i}^{2})\\
&=& \alpha_j^i P(\tilde{y}_i^k(d)|(\mu_{jk})_{i},(\sigma_{jk})_{i}^{2})
\end{eqnarray*}

In the expectation step of EM, we compute the expectation of the log-likelihood function with respect to the hidden variables $z_j$, given the measurements $\mathbf{\tilde{Y}}_i$, and the current estimate of the parameter set $\theta^l$. This is given by

\begin{eqnarray*}
Q(\theta,\theta^l)&=& E[\log P(\mathbf{\tilde{Y}}_i,Z|\theta)|\mathbf{\tilde{Y}}_i,\theta^l]\\
&=& \sum\limits_{j=1}^{2}\sum\limits_{d=1}^{D_1(t)}\log[\alpha_j^i P(\tilde{y}_i^0(d)|(\mu_{j0})_{i},(\sigma_{j0})_{i}^{2})P(z_j|\tilde{y}_i^0(d),\theta^l)]\\
&+& \sum\limits_{j=1}^{2}\sum\limits_{d=D_1(t)+1}^{D}\log[\alpha_j^i P(\tilde{y}_i^1(d)|(\mu_{j1})_{i},(\sigma_{j1})_{i}^{2})P(z_j|\tilde{y}_i^1(d),\theta^l)]
\end{eqnarray*}

where 

\begin{equation}
P(z_j|\tilde{y}_i^k(d),\theta^l)=\frac{\alpha_j^i(l) P(\tilde{y}_i^k(d)|(\mu_{jk})_{i}(l),(\sigma_{jk})_{i}^{2}(l))}{\sum\limits_{n=1}^{2}\alpha_n^i(l) P(\tilde{y}_i^k(d)|(\mu_{nk})_{i}(l),(\sigma_{nk})_{i}^{2}(l))}.
\end{equation}

In the maximization step of EM algorithm, we maximize
$Q(\theta,\theta^l)$ with respect to the parameter set $\theta$ so as to compute the next parameter set: $$\theta^{l+1}=\underset{\theta}{\arg\max}\;\; Q(\theta,\theta^l).$$
First, we maximize $Q(\theta,\theta^l)$ subject to the constraint $\left(\sum\limits_{j=1}^{2}\alpha_j^i=1\right)$. We define the Lagrangian $\mathcal{L}$ as
$$\mathcal{L}=Q(\theta,\theta^l)+\lambda\{\sum\limits_{j=1}^{2}\alpha_j^i-1\}.$$ Now, we equate the derivative of $\mathcal{L}$ to zero: $$\frac{d}{d\alpha_j^i}\mathcal{L}=\lambda+\frac{\sum\limits_{d=1}^{D_1(t)}P(z_j|\tilde{y}_i^0(d),\theta^l)}{\alpha_j^i}+\frac{\sum\limits_{d=D_1(t)+1}^{D}P(z_j|\tilde{y}_i^1(d),\theta^l)}{\alpha_j^i}=0.$$
Multiplying both sides by $\alpha_j^i$ and summing over $j$ gives $\lambda=-D$. Similarly, we equate the derivative of $Q(\theta,\theta^l)$ with respect to $(\mu_{jk})_{i}$ and $(\sigma_{k})_{i}^{2}$ to zero. Now, an iterative algorithm for all the parameters is

\begin{eqnarray}
&&
\alpha_j^i(l+1)=\frac{1}{D}\left[\sum\limits_{d=1}^{D_1(t)}P(z_j|\tilde{y}_i^0(d),\theta^l)+\sum\limits_{d=D_1(t)+1}^{D}P(z_j|\tilde{y}_i^1(d),\theta^l)\right]\\
&&
(\mu_{j0})_{i}(l+1)= \frac{\sum\limits_{d=1}^{D_1(t)}P(z_j|\tilde{y}_i^0(d),\theta^l) \tilde{y}_i^0(d)}{\sum\limits_{d=1}^{D_1(t)}P(z_j|\tilde{y}_i^0(d),\theta^l)}\\
&&
(\mu_{j1})_{i}(l+1)= \frac{\sum\limits_{d=D_1(t)+1}^{D}P(z_j|\tilde{y}_i^1(d),\theta^l)\tilde{y}_i^1(d)}{\sum\limits_{d=D_1(t)+1}^{D}P(z_j|\tilde{y}_i^1(d),\theta^l)}\\
&&
(\sigma_{j0})_{i}^{2}(l+1)=\frac{\sum\limits_{j=1}^{2}\sum\limits_{d=1}^{D_1(t)}P(z_j|\tilde{y}_i^0(d),\theta^l) (\tilde{y}_i^0(d)-(\mu_{j0})_{i}(l+1))^2}{\sum\limits_{j=1}^{2}\sum\limits_{d=1}^{D_1(t)}P(z_j|\tilde{y}_i^0(d),\theta^l)}\label{s0}\\
&&
(\sigma_{j1})_{i}^{2}(l+1)=\frac{\sum\limits_{j=1}^{2}\sum\limits_{d=D_1(t)+1}^{D}P(z_j|\tilde{y}_i^1(d),\theta^l) (\tilde{y}_i^1(d)-(\mu_{j1})_{i}(l+1))^2}{\sum\limits_{j=1}^{2}\sum\limits_{d=D_1(t)+1}^{D}P(z_j|\tilde{y}_i^1(d),\theta^l)}\label{s1}
\end{eqnarray}

In the learning iteration $t$, let the estimates after the convergence of the above algorithm be denoted by $\hat{\theta}(t)=\{\hat{\alpha}_j^i(t),(\hat{\mu}_{jk})_{i}(t),(\hat{\sigma}_{jk})_{i}^{2}(t)\}$. These estimates are then used as the initial values for the next learning iteration $t+1$ that uses a new set of $D$ data samples.

After learning the unknown parameter set under $I_0$ and $I_1$, node classification can be achieved using the following maximum likelihood decision rule:
 
\begin{equation}
\hat{f}(\mathbf{Y}_i|\;I_0) \quad \mathop{\stackrel{H}{\gtrless}}_{B} \quad  \hat{f}(\mathbf{Y}_i|\;I_1)
\end{equation}

where $\hat{f}(\cdot)$ is the PDF based on estimated parameters.

Using the above estimates and node classification, the optimal distributed weights for honest nodes after learning iteration $t$ can be written as

\begin{equation}
w_i^H(t)=\frac{(\hat{\mu}_{11})_{i}(t)-(\hat{\mu}_{10})_{i}(t)}{(\hat{\sigma}_{10})_{i}^{2}(t)}.
\end{equation}

Similarly, the optimal distributed weights for Byzantines after learning iteration $t$ can be written as

\begin{equation}
w_i^B(t)=\frac{\sum\limits_{j=1}^{2}\hat{\alpha}_j^i(t)[(\mu_{j1})_{i}(t)-(\hat{\mu}_{j0})_{i}(t)]}
{\hat{\alpha}_1^i(t)\hat{\alpha}_2^i(t)\left(\left(\hat{\mu}_{10}(t)\right)_i-\left(\hat{\mu}_{20}(t)\right)_i\right)^2+\left(\hat{\alpha}_1^i(t)\left(\hat{\sigma}_{10}\right)^2_i(t)+\hat{\alpha}_2^i(t)\left(\hat{\sigma}_{20}\right)^2_i(t)\right)}
\end{equation}

\begin{figure}[t] 
  \centering
    \includegraphics[width=0.5\textwidth]{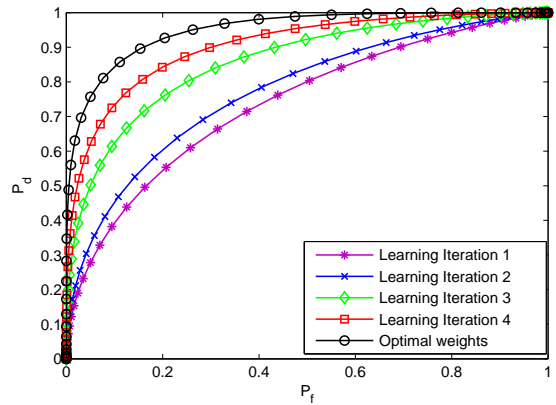}
    \caption{ROC for different learning iterations}
    \label{learn1}
\end{figure}

Next, we present some numerical results in Figure~\ref{learn1} to evaluate the performance of our proposed scheme.
Consider the scenario where $6$
nodes organized in an undirected graph (as shown in Figure~\ref{fig1})
are trying to detect a phenomenon. Node $1$ and node $2$ are considered to be Byzantines.
We assume that $((\mu_{10})_{i},(\sigma_{10})_{i}^{2})=(3,1.5)$,
$((\mu_{11})_{i},(\sigma_{11})_{i}^{2})=(4,2)$ and the attack parameters are $(P_i,\Delta_i)=(0.5,9)$. In Figure~\ref{learn1}, we plot ROC curves for different learning iterations. For every learning iteration, we assume that $D_1=10$ and $D=20$. It can be seen from Figure~\ref{learn1} that within $4$ learning iterations, detection performance of the learning based weighted gain combining scheme approaches the detection performance of weighted gain combining with known optimal weight based scheme.     

Note that, the above learning based scheme can be used in conjunction with the proposed weighted average consensus based algorithm to mitigate the effect of Byzantines.

\section{Conclusion and future work}
\label{conclusion}
In this paper, we analyzed the security performance of conventional consensus based algorithms in the presence of data falsification attacks. 
We showed that above a certain fraction of Byzantine attackers in the network, existing consensus based detection algorithm are ineffective.
Next, we proposed a robust distributed weighted average consensus algorithm and devised a learning technique to estimate the operating parameters (or weights) of the nodes. This enables an adaptive design of the local fusion or update rules to mitigate the effect of data falsification attacks. 
There are still many interesting questions that remain to be explored in the future work such as an analysis of the problem for time varying topologies. Note that, some analytical methodologies used in this paper are certainly exploitable for studying the attacks in time varying topologies. Other questions such as the optimal topology which incurs the fastest convergence rate can also be investigated.

\bibliographystyle{IEEEtran}
\bibliography{Book,Journal}

\end{document}